\documentclass[a4paper,USenglish, numberwithinsect]{lipics-v2019}

\usepackage{microtype}
\usepackage[defaultlines=4,all]{nowidow}
\title{Recursive Backdoors for SAT} 
\titlerunning{Recursive Backdoors for SAT} 

\author{Nikolas Mählmann}{University of Bremen, Germany}{maehlmann@uni-bremen.de}{}{}

\author{Sebastian Siebertz}{University of Bremen, Germany}{siebertz@uni-bremen.de}{https://orcid.org/0000-0002-6347-1198}{}

\author{Alexandre Vigny}{University of Bremen, Germany}{vigny@uni-bremen.de}{https://orcid.org/0000-0002-4298-8876}{}

\authorrunning{N.\ Mählmann, S.\ Siebertz, A.\ Vigny} 

\Copyright{Nikolas Mählmann, Sebastian Siebertz, Alexandre Vigny}  

\ccsdesc[500]{Theory of computation~Logic}
\ccsdesc[500]{Theory of computation~Parameterized complexity and exact algorithms}

\keywords{Propositional satisfiability SAT, Backdoors, Parameterized Algorithms.}

\nolinenumbers 

\hideLIPIcs  

\EventEditors{John Q. Open and Joan R. Access}
\EventNoEds{2}
\EventLongTitle{42nd Conference on Very Important Topics (CVIT 2016)}
\EventShortTitle{CVIT 2016}
\EventAcronym{CVIT}
\EventYear{2016}
\EventDate{December 24--27, 2016}
\EventLocation{Little Whinging, United Kingdom}
\EventLogo{}
\SeriesVolume{42}
\ArticleNo{}

\usepackage{amsthm} 
\usepackage{mathtools}
\usepackage{graphicx}
\usepackage{mathrsfs}
\usepackage{xcolor}
\usepackage{amssymb}
\usepackage{mathabx}
\usepackage{tikz}
\usepackage{float}
\usepackage{wrapfig}
\usepackage{xspace}
\definecolor{cadmiumgreen}{rgb}{0.0, 0.42, 0.24}
\definecolor{dark-blue}{rgb}{0.05,0.25,1}
\usepackage{hyperref}
\hypersetup{
    colorlinks = true,
    linkcolor = blue,
    citecolor=blue
    }
\usepackage[noabbrev,capitalise,nameinlink]{cleveref}
\usepackage{wrapfig}
\usepackage{tikzit}

\tikzstyle{tz_vertex}=[fill=white, draw=black, shape=circle]
\tikzstyle{tz_blue_vertex}=[fill={rgb,255: red,137; green,207; blue,240}, draw=black, shape=circle]
\tikzstyle{tz_orange_vertex}=[fill={rgb,255: red,250; green,218; blue,94}, draw=black, shape=circle]
\tikzstyle{tz_vertex_green}=[fill={rgb,255: red,172; green,225; blue,175}, draw=black, shape=circle]
\tikzstyle{tz_red}=[fill={rgb,255: red,254; green,111; blue,94}, draw=black, shape=circle]
\tikzstyle{tz_big_text}=[fill=none, draw=none, shape=circle, font={\LARGE}]
\tikzstyle{tz_medium_text}=[fill=none, draw=none, shape=circle, font={\Large}]

\tikzstyle{tz_dashed}=[-, dashed]
\tikzstyle{tz_blue_bg}=[-, fill={rgb,255: red,137; green,207; blue,240}, dashed]
\tikzstyle{tz_green}=[-, dashed, draw=black, fill={rgb,255: red,172; green,225; blue,175}]
\tikzstyle{tz_orange}=[-, fill={rgb,255: red,250; green,218; blue,94}]
\tikzstyle{tz_dark_green_edge}=[-, fill=none, draw={rgb,255: red,172; green,225; blue,175}, ultra thick]
\tikzstyle{tz_dark_green_edge_dashed}=[-, draw={rgb,255: red,172; green,225; blue,175}, dashed, ultra thick]
\tikzstyle{tz_dark_blue_edge}=[-, draw={rgb,255: red,137; green,207; blue,240}, ultra thick]
\tikzstyle{tz_dark_blue_edge_dashed}=[-, draw={rgb,255: red,137; green,207; blue,240}, dashed, ultra thick]

\usetikzlibrary{patterns,arrows,decorations.pathreplacing}

\newcommand{\Cc}{\mathscr{C}}

\newcommand{\Bb}{\mathcal{B}}
\newcommand{\Tt}{\mathcal{T}}

\renewcommand{\phi}{\varphi}

\definecolor{unverifiedBlue}{rgb}{0.16, 0.32, 0.75}
\newenvironment{unverified}{\begingroup \color{unverifiedBlue}}{\endgroup}

\newcommand{\bd}{\mathrm{bd}}
\newcommand{\wbd}{\mathrm{wrbd}}
\newcommand{\sbd}{\mathrm{srbd}}

\newcommand{\sbdtwz}{\sbd_{\Cc_{0}}}

\newcommand{\lk}{\lambda_k}
\newcommand{\nd}[2]{N^\dagger_{#1}[#2]}

\newcommand{\inductionBaseCase}[0]{\medskip \noindent \emph{Base Case: }}

\newcommand{\inductionStep}[0]{\medskip \noindent \emph{Induction Step: }}
\newcommand{\algCase}[1]{\medskip \noindent \emph{Case #1: }}
\newcommand{\algHeading}[1]{\medskip \noindent \emph{#1:}}
\newcommand{\bigO}[1]{\ensuremath{\mathcal{O}(#1)}}

\newcommand{\recursiveBackdoor}{recursive backdoor\xspace}
\newcommand{\recursiveBackdoors}{recursive backdoors\xspace}

\newcommand{\srb}{\ensuremath{\mathrm{SRB}}}
\newcommand{\SRB}{\ensuremath{\mathrm{SRB}}\xspace}
\newcommand{\defSBDCase}{\parbox[c]{.18\textwidth}{\smallskip
    \baselineskip10pt
if $G\not\in \Cc$ and $G$ is connected\smallskip} }

\newcommand{\var}{\mathrm{var}}
\newcommand{\cla}{\mathrm{cla}}
\newcommand{\lit}{\mathrm{lit}}
\newcommand{\el}{V}

\begin{document}
\maketitle              

\begin{abstract}
  A strong backdoor in a formula $\phi$ of propositional logic to a
  tractable class $\Cc$ of formulas is a set~$B$ of variables of
  $\phi$ such that every assignment of the variables in $B$ results in
  a formula from $\Cc$.  Strong backdoors of small size or with a good
  structure, e.g.\ with small backdoor treewidth, lead to efficient
  solutions for the propositional satisfiability problem SAT.

  In this
  paper we propose the new notion of \emph{recursive backdoors}, which
  is inspired by the observation that in order to solve SAT we can
  independently recurse into the components that are created
  by partial assignments of variables. The quality of a recursive backdoor is measured
  by its \emph{recursive backdoor depth}.  Similar to the concept of
  backdoor treewidth, recursive backdoors of bounded depth include
  backdoors of unbounded size that have a certain treelike
  structure. However, the two concepts are incomparable and our
  results yield new tractability results for SAT.
\end{abstract}



\section{Introduction}

The problem of checking whether a formula of propositional logic in
conjunctive normal form~(CNF) is satisfiable (SAT) is one of the most
central problems in computer science. The problem is often seen as the
canonical NP-complete problem~\cite{cook1971complexity} and
conjectured to be not solvable in sub-exponential
time~\cite{impagliazzo2001problems}. Despite this theoretical hardness
result, state-of-the-art SAT solvers are able to efficiently solve
multi-million variable instances arising from real-world
applications. We refer to the recent survey of Vardi and
Ganesh~\cite{GaneshV20}, who try to explain this ``unreasonable
effectiveness of SAT solvers''. SAT is known to be solvable in
polynomial time on several restricted classes of formulas, e.g.\ on
Horn and 2CNF formulas.  However, this classification falls short of
explaining the practical efficiency of SAT solvers, as many
efficiently solvable instances do not belong to any of these classes.

Parameterized complexity theory offers a refined view on the
complexity of problems. Instead of measuring complexity only with
respect to the input size $n$, one or more parameters are taken into
account. Optimally, one can establish fixed-parameter tractability
with respect to a parameter $k$, that is, a running time of
$f(k)\cdot n^c$ for some computable function $f$ and a constant
$c$. In case the parameter $k$ is small on a given class of instances,
this may lead to efficient algorithms even if the inputs are
large. Even though SAT solvers may not be explicitly tailored to use
these parameters, it is conceivable that they implicitly exploit the
structure that is imposed by them.  This poses the question of
parametric characterizations of real-world application instances that
can be solved efficiently. One very important parameter to explain
tractability is \emph{treewidth}, which intuitively measures how
tree-like an instance is, and which can be used to obtain
fixed-parameter tractability for SAT~\cite{samer2010algorithms}.

A second very successful parametric approach was introduced by
Williams et al.~\cite{williams2003backdoors}. For a formula $\phi$, a
\emph{strong backdoor} to a given class $\Cc$ of formulas is a set of
variables of $\phi$ such that for every assignment of these variables
one obtains a formula in $\Cc$. Similarly, a \emph{weak backdoor} to
$\Cc$ for a satisfiable formula is a set of variables of $\phi$ such
that some assignment of these variables leads to a formula in
$\Cc$. These notions elegantly allow to lift tractability results from
classes $\Cc$ to classes that are \emph{close} to $\Cc$. Given a
formula and a strong backdoor of size $k$ to a tractable class, one
can decide satisfiability by checking $2^k$ tractable instances. For
small $k$ this yields efficient algorithms as noted by Nishimura et
al.~\cite{nishimura2004detection}, who first studied the parameterized
complexity of backdoor detection.

A lot of effort has been invested to develop fpt algorithms for
backdoor detection to various tractable base classes $\Cc$, for
example to classes of bounded treewidth
\cite{gaspers2013backdoorstotw} or heterogeneous classes
\cite{gaspers2017heterogenous}.  Treewidth is a width measure for
graphs that can however be applied to measure the complexity of
formulas by considering the incidence graphs of formulas. The
incidence graph of a formula has one vertex for each variable and one
vertex for each clause.  A variable vertex is connected with a clause
vertex when the variable is contained positively or negatively in the
clause. In the following, we will often use graph theoretic
terminology for formulas, and this always refers to the incidence
graph of the formula.

Apart from various base classes, alternative measures of quality of
backdoors have been proposed.  Backdoor trees generalize backdoor sets
into decision trees, whose quality is measured by their number of
leafs \cite{samer2008backdoortrees}.  Recently backdoor trees have
been further generalized to backdoor DNFs~\cite{ordyniak2021dfns}.
Ganian et al.~\cite{ganian2017backdoor} introduced the notion of
backdoor treewidth, which permits fpt backdoor detection for backdoors
of unbounded size. Even though backdoors of bounded treewidth can be
arbitrarily large, they showed that SAT is fixed-parameter tractable
when parameterized by the backdoor treewidth with respect to the
classes $\Cc$ of Horn, Anti-Horn and 2CNF formulas. They also consider
backdoors that split the input CNF formula into components that may
belong to different tractable classes.  For a an overview of
additional works we refer to the survey by Gaspers and Szeider
\cite{gaspers2012backdoors} as well as to the upcoming book chapter by
Samer and Szeider \cite{samer2021fpt}.

In this paper we introduce the new notions of \emph{strong} and
\emph{weak \recursiveBackdoors} as generalizations of backdoor sets
and backdoor trees.  Strong \recursiveBackdoors extend backdoor trees
by not only branching on truth values but also recursively branching
into the independent components of the formula that may arise after
the partial assignment of variables. We measure the quality of
\recursiveBackdoors by the depth of their branching trees.  The
splitting into components allows \recursiveBackdoors of bounded depth
to contain an unbounded number of variables.  Our definition, together
with the observation that after the assignment of a variable one can
independently solve the sub-instances in the arising components,
reveals a new potential of backdoors for SAT.

The main power of recursive backdoors, but also the difficulty in
their study, is that by assigning a variable we do not recurse into
the components that are created by deleting that variable, but into
the components that are created by deleting parts of the neighborhood
of the variable. We show that detecting weak recursive backdoors even
to the class $\Cc_0$ of edgeless incidence graphs is W[2]-hard.  Our
main technical contribution is an fpt algorithm that, given a formula
$\phi$ and a parameter $k$, either decides satisfiability of $\phi$ or
correctly concludes that $\phi$ has no strong \recursiveBackdoor to
$\Cc_0$ of depth at most $k$.  Even for the class $\Cc_0$ this yields
tractability results that cannot be achieved by backdoor treewidth.

 \medskip

 We provide background in \Cref{sec:prelims}. We define
 recursive backdoors in \Cref{sec:backdoors} and \Cref{sec:sketch}
 is devoted to a sketch of the fpt algorithm. The rest of the paper is
 devoted to the formal presentation and correctness proof of that
 algorithm.  Due to space constraints we present the hardness proof
 only in the appended full version of the paper.  Also some proofs of
 the main result are deferred to the appendix.


\section{Preliminaries}\label{sec:prelims}


\textbf{Propositional Logic.} We consider formulas of
propositional logic in conjunctive normal form (CNF), represented
by finite sets of clauses, and in the following when we speak of a
formula we will always mean a CNF formula.
We write $x,y,z\ldots$ for variables
and $\star,\diamond \in \{+,-\}$ for polarities.
A literal is a variable with an assigned polarity.
We write~$x_+$ for the positive literal $x$,
$x_-$ for the negative literal $\bar x$,
and $x_\star$ for a literal with arbitrary polarity.
Every clause is a finite set of literals.
We assume that no clause contains a complementary pair $x_+, x_-$.
For a formula $\phi$,
we write $\var(\phi)$ and $\cla(\phi)$ to refer to the sets
of variables and clauses of $\phi$, respectively. We say that a
variable~$x$ is positive (resp.\ negative) in a clause $c$ if $x_+\in c$
(resp.\ $x_- \in c$), and we write $\var(c)$
for the set of variables $x$ with $x_\star\in c$ and
$\lit(c)$ for the set of literals in $c$. For a formula $\phi$
we let $\var(\phi)=\bigcup_{c\in \phi}\var(c)$.

The {\em width} of $c$ is $|\var(c)|$ and the {\em length} of $\phi$ is
$\sum_{c\in \phi}|\var(c)|$, denoted~$|c|$ and $|\phi|$,
respectively.
We call a clause a {\em $d$-clause} if it has width exactly $d$.
We write $\Cc_d$ to refer to the class of CNF formulas whose clauses have at most width $d$.
Especially $\Cc_0$ denotes the class of empty formulas,
that either contain empty clauses or no clauses at all.

A {\em truth assignment} $\tau$ is a mapping from a set of variables,
denoted by $\var(\tau)$, to $\{+,-\}$.
A truth assignment $\tau$
satisfies a clause $c$ if $c$ contains at least one literal $x_\star$
with $\tau(x)=\star$.
A truth assignment $\tau$ of $\var(\phi)$ satisfies
the formula $\phi$ if it satisfies all clauses of $\phi$.

Given a formula $\phi$ and a truth assignment $\tau$,
$\phi[\tau]$ denotes the formula obtained from~$\phi$
by removing all clauses that are satisfied by $\tau$ and by
removing from the remaining clauses all literals $x_\star$ with
$\tau(x)\neq \star$. Note that for every formula $\phi$ and
assignment $\tau$ we have
$\var(\phi[\tau])\cap \var(\tau)=\emptyset$. If $\tau$ and
$\tau'$ are assignments with $\var(\tau)\cap \var(\tau')=\emptyset$,
then we write $\tau\cup\tau'$ for the unique assignment extending
both $\tau$ and $\tau'$.

\medskip\noindent
\textbf{Graphs.} We will consider only graphs that arise as
incidence graphs of formulas. The incidence graph
$G_\phi$ of a formula $\phi$ is a bipartite graph
with vertices $\var(\phi)\cup \cla(\phi)$. Slightly abusing
notation we usually do not distinguish between a formula and its
incidence graph. E.g.\ we speak of the variables and clauses
of $G_\phi$, which we denote by $\var(G_\phi)$ and $\cla(G_\phi)$
respectively. Vice versa, we speak e.g.\ of components of $\phi$
with implicit reference to the incidence graph $G_\phi$.
We drop the subscript $\phi$ if it is clear from
the context. The edges of $G$ are partitioned
into two parts $E_+$ (positive edges) and $E_-$ (negative edges),
where a variable~$x$ is connected to a clause $c$ by
an edge $E_\star$ if $x_\star\in \lit(c)$.
For an assignment $\tau$ we naturally define
$G[\tau]$ as the incidence graph of $\phi[\tau]$.
If $\tau$ assigns only a single variable
$x \mapsto \star$ we write $G[x_\star]$ for $G[\tau]$.
Note that for every assignment $\tau$, $G[\tau]$ is an induced subgraph of~$G$.
For a vertex $v$ the closed $\star$-neighborhood of $v$ is defined as
$N_\star[v]\coloneqq \{w~:~\{v,w\}\in E_\star\}$. For $W\subseteq V$
we write $G[W]$ for the subgraph induced by $W$ and
$G-W$ for $G[V\setminus W]$.

We refrain from formally defining treewidth and
refer to the literature for background.
A graph $H$ is a minor of a
graph $G$ if $H$ can be obtained from $G$ by deleting edges
and vertices and by contracting edges.
To compare our new definition of recursive backdoors with backdoor treewidth,
we mention that if a graph contains a $k\times k$ grid as a minor, then it
has treewidth at least $k$~\cite{robertson1986graph}.

\medskip\noindent
\textbf{Parameterized Complexity.}
A parameterized problem is called fixed-parameter tractable (fpt) if there
exists an algorithm deciding the problem in time $f(k)\cdot n^c$,
where $n$ is the input size, $k$ is the parameter, $f$ is a computable
function and $c$ is a constant. An algorithm witnessing fixed-parameter
tractability of a problem is called an fpt-algorithm for the problem.

To show that a problem is likely to be not fpt one can show that it is
W[$i$]-hard for some~$i\geq 1$. For this, it is sufficient to give a
parameterized reduction from a known W[$i$]-hard problem. We
refer to the book~\cite{cygan2015parameterized} for extensive background on parameterized
complexity theory.

\medskip\noindent
\textbf{Backdoors.} Let $\Cc$ be a class of formulas and 
let $\phi$ be a formula. 
A set $B\subseteq \var(\phi)$ is a  \emph{strong backdoor} of $\phi$ to $\Cc$ if 
for every assignment $\tau\colon B\rightarrow\{+,-\}$ the formula~$\phi[\tau]$ belongs to~$\Cc$. Note that for some assignments $\tau$ the formula $\phi[\tau]$
may not be satisfiable, even though~$\phi$ is satisfiable. Hence, in the 
following definition of a weak backdoor we require that~$\phi$ is satisfiable. 
If $\phi$ is satisfiable, then
a set $B\subseteq \var(\phi)$ is a \emph{weak backdoor} of~$\phi$ to the
class $\Cc$ if there exists an assignment $\tau\colon B\rightarrow\{+,-\}$
such that $\phi[\tau]$ is a satisfiable formula in $\Cc$. 
The classical measure for the complexity or quality of a backdoor is its size.

An important recent approach to measure the complexity of a backdoor
is to take its structure into account. The \emph{treewidth of a
backdoor} $B$ is defined as the treewidth of the graph with vertex set $B$
where two variables $x$ and $y$ are connected by an edge if there exists a
path from a neighbor of~$x$ to a neighbor of $y$ in $G-B$~\cite{ganian2017backdoor}.
Ganian et al.~\cite{ganian2017backdoor} also consider
backdoors that split the input CNF formula into
components that each may belong to a different tractable
class~$\Cc$.

\medskip\noindent
\textbf{Permissive Backdoor Detection.}    
In their survey Gaspers and Szeider~\cite{gaspers2012backdoors} differ
between a strict and a permissive version of the
backdoor detection problem.
Given a backdoor definition $\Bb$ and a corresponding quality measure $\mu$ 
(e.g. strong backdoors to 2CNF measured by their size) as
well as a formula $\phi$ and a parameter $k$.
The strict backdoor detection problem, denoted as $\Bb$-$\textsc{Detection}$, asks 
whether or not $\mu(\phi) \leq k$ holds.
The permissive backdoor detection problem, denoted as $\textsc{SAT}(\mu)$, asks
to either decide the satisfiability of $\phi$ or conclude that $\mu(\phi) > k$ holds.
The permissive version of the problem grants more freedom in algorithm design,
as trivial instances can be solved without calculating the backdoor measure.
However as Gaspers and Szeider point out, 
hardness proofs seem to be much harder for the permissive version.

\section{Recursive Backdoors}\label{sec:backdoors}
\textbf{Strong Recursive Backdoors.}
Our new concept of recursive backdoors is based on the observation
that we can handle the components of $G[x_\star]$
independently whenever
a variable~$x$ has been assigned.
A \emph{strong \recursiveBackdoor} (\srb) of an incidence graph $G$ to a
class $\Cc$  is a rooted
labeled tree, where every node is either labeled with a subgraph of
$G$ or with a variable in $\var(G)$. The root of the tree is
labeled with $G$. Whenever
an inner node is labeled with a connected graph $H$, then
it has one child labeled with a variable. Whenever it is labeled
with a disconnected graph, then it has one child for each of
its components, labeled with the graph induced by that
component. Whenever an inner node is labeled with a variable
$x$, then its parent is labeled with a graph $H$, and its
two children are labeled with $H[x_+]$ and $H[x_-]$,
respectively. Every leaf node is labeled with a graph from
$\Cc$. We call the nodes of the tree 
{\em variable nodes} or {\em component nodes},
according to to their labeling.

The {\em depth} of a strong recursive backdoor
is the maximal number of variable nodes
from its root to one of its leafs.
The \emph{strong \recursiveBackdoor depth} to a class $\Cc$ ($\sbd_{\Cc}$)
of an incidence graph~$G$
is the minimal depth of a strong recursive backdoor of $G$ to $\Cc$.
We give the following equivalent definition:

\begin{definition}[Strong Recursive Backdoor Depth]\label[definition]{def_sbd}
    \[
    \sbd_\Cc(G) =
    \begin{cases}
    0 & \text{if $G\in \Cc$} \\
    1 + \min_{x\in \var(G)}
        \max_{\star \in \{+,-\}}\,
            \sbd_\Cc(G[x_\star])
    & \defSBDCase \\
    \max\,
    \{\,\sbd_\Cc(H)
    ~:~H\text{ connected component of $G$\,}\}
    &\text{otherwise}
    \end{cases}
    \]
\end{definition}

\smallskip
To get a better understanding of strong recursive backdoor
depth we give an example of a family of incidence graphs
with unbounded backdoor treewidth to 2CNF
but constant strong recursive backdoor depth to
$\Cc_0$, the class of edgeless graphs.
For any $k\geq 0$, define the graph~$G_k$
as follows. We start with a $k\times k$ grid of clause
vertices $\{c_{1,1},...,c_{k,k}\}$,
depicted in yellow in \cref{fig:sbd_and_backdoor_treewidth}. We connect a private
variable vertex to each of the corners $c_{1,1},c_{1,k},c_{k,1},c_{k,k}$ 
of the grid. We now replace each
path of the grid by a path of length $6$ (containing $5$
vertices). Every second vertex on a new path is a
clause vertex, connected to the two adjacent variable
vertices. Furthermore, we add a special variable vertex $x$
that is connected alternatingly with positive and negative
polarity (depicted in green and blue in the figure)
to the clause vertices on the new paths. Variable vertices
are depicted as white vertices in the figure.

Since every clause
$c_{i,j}$ is connected to at least $3$ variables,
every backdoor set $B$ to 2CNF
will have to contain at least one variable of every $c_{i,j}$.
No matter which other variables are picked into $B$,
the backdoor $B$-torso graph of $G$ will always contain
a $k \times k$ grid as a minor, hence, will have
treewidth at least $k$.

A strong recursive backdoor to $\Cc_0$ with $x$ as its root
splits every grid clause into a separate component
of constant size and and therefore has constant depth.

\begin{figure}[H]
    \scalebox{0.47}{\input{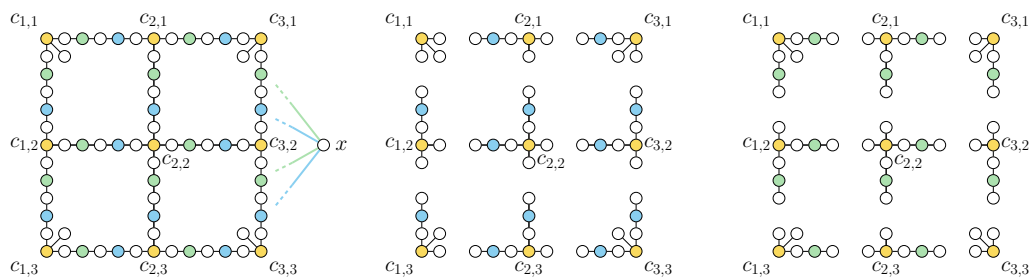}}
    \caption{From left to right: $G_3$, $G_3[x_+]$, and $G_3[x_-]$.}
    \label[figure]{fig:sbd_and_backdoor_treewidth}
\end{figure}

\medskip\noindent
\textbf{SAT Solving and SAT Counting using Strong Recursive Backdoors.}
Similar to regular strong backdoor sets,
strong recursive backdoors allow for polynomial time SAT Solving and SAT Counting
if the backdoor is given as part of the input.
First, we observe that even though it may have an unbounded branching degree,
the size of a recursive backdoor is still linear in the size of its formula.

\begin{lemma}\label[lemma]{srbs_have_few_leafs}
    Let $T$ be a strong recursive backdoor with depth $k$ of a formula $\phi$ to a
    class~$\Cc$.
    The number of leaf nodes in $T$,
    as well as the sum of number of vertices contained in leaf nodes
    is bound by $2^k \cdot |\phi|$.
\end{lemma}

\begin{proof}
    Proof by induction on $k$ and $|\phi|$.
    The bound trivially holds when $k=0$ or $|\phi|=1$
    and the backdoor consists of a single leaf node.

    In the inductive step,
    a variable node increases the backdoor depth
    and at most doubles the number of leaves and their contained vertices,
    as it branches on both polarities.
    A component node does not increase the backdoor depth,
    but branches over disjoint components of strictly smaller size.
    As the sum of the vertices contained in the components is equal to $|\phi|$,
    the number of leaves and contained vertices is again bounded by~$2^k \cdot |\phi|$.
\end{proof}
%
%
%
We can use this observation to construct a straight-forward
bottom-up algorithm:

\begin{proposition}\label[proposition]{sat_using_srbds}
    Given a strong recursive backdoor with depth $k$ of a formula $\phi$ to a tractable class $\Cc$,
    we can check the satisfiability or count satisfying assignments for $\phi$
    in time~$2^k \cdot \mathrm{poly}(|\phi|)$.
\end{proposition}

\begin{proof}
    By \cref{srbs_have_few_leafs}, we know that we have at most $2^k \cdot |\phi|$ instances labeling leaves,
    which can be solved in polynomial time.
    For variable nodes, the instance labeling the node is satisfiable if and only if
    at least one of its two children is labeled with a satisfiable instance.
    The number of satisfiable assignments is the sum of the satisfiable
    assignments for its children.
    For component nodes, the instance labeling the node is satisfiable if and only if
    all of its children are labeled with satisfiable instances.
    The number of satisfiable assignments is the product of the satisfiable assignments for its children.
\end{proof}


\medskip\noindent
\textbf{Weak Recursive Backdoors.}
Recall that in the definition of weak backdoors we consider
only satisfiable formulas and aim to find an assignment
$\tau$ that leads to a satisfiable formula $\phi[\tau]\in \Cc$.
This is also the case in the following definition of weak recursive
backdoor depth:

\begin{definition}[Weak Recursive Backdoor Depth]\label[definition]{def_wbd}
    \[
    \wbd_\Cc(G) =
    \begin{cases}
    0 & \parbox[c]{.18\textwidth}{\smallskip\baselineskip10pt
        if $G\in \Cc$ and $G$ is satisfiable\smallskip}\\

    \infty & \parbox[c]{.18\textwidth}{\smallskip\baselineskip10pt
        if $G\in \Cc$ and $G$ is unsatisfiable\smallskip}\\

    1 + \min_{x\in \var(G)}
        \min_{\star \in \{+,-\}}\,
            \wbd_\Cc(G[x_\star])
    & \defSBDCase \\
    \max\,
    \{\,\wbd_\Cc(H)
    ~:~H\text{ connected component of $G$\,}\}
    &\text{otherwise}
    \end{cases}
    \]
\end{definition}

\medskip
\noindent
\textbf{SAT Solving using Weak Recursive Backdoors.}
We can use the notion of weak recursive backdoors for SAT as follows.
This time we do not assume that the backdoor is given with the input. Let
$\Cc$ be a class such that we can test membership
and satisfiability in polynomial time. As above it is easy to see that the straight-forward algorithm that recursively assigns variables yields the
following result. Given a formula $\phi$ and an integer~$k$, we can
test in time $(2|\phi|)^k\cdot poly(|\phi|)$, whether $\phi$ is satisfiable,
unsatisfiable, or $\phi$ has weak recursive backdoor depth to $\Cc$
greater than $k$.
Note that when $k$ is small, even this running time is a major
improvement over the worst case running time of $2^{cn}$ implied by
the exponential time hypothesis (ETH).

%




\section{Proof Sketch}\label{sec:sketch}
Our goal is to show that the permissive backdoor detection problem
$\textsc{SAT}(\sbdtwz)$ is fixed-parameter
tractable. That is, we aim to decide for a given formula $\phi$ and
parameter $k$ whether $\phi$ is satisfiable or does not have a
strong recursive backdoor of depth $k$ to the class~$\Cc_0$ of empty formulas
i.e.~$\sbdtwz(\phi) > k$.
Our approach is based on two main
observations:
Formulas with strong recursive backdoor depth $k$ to
$\Cc_0$ have both a maximal clause degree $k$ (see \cref{limited_clause_degree})
and a diameter bounded by $\lk :=4\cdot 2^k$ in each connected component (see \cref{low_diameter}).

We are going to design a recursive algorithm that in every step
finds a bounded depth~\SRB
that reduces the maximal clause degree of $\phi$ by one,
or proves that the strong recursive
backdoor depth of $\phi$ is larger than $k$.
We extend the \SRB by recursively branching on its leaves until we reach $\Cc_0$.
Since $\phi$ has maximal clause degree $k$ this 
yields a bounded depth \SRB to $\Cc_0$ in fpt running time.


First, take a look at the special case where $G := G_\phi$ contains a clause  $c$ of width $k$, i.e.\ a $k$-clause.
By our first observation, the existence of $c$ implies that $\sbdtwz(G) \geq k$.
Note that the degree of $c$ can only be reduced by assigning a variable in its neighborhood.

Next, consider the case where $G$ contains two disjoint
$(k-1)$-clauses $c_1, c_2$ in the same connected component.
Since $G$ has limited diameter,
there exists a path $P$ of length $\leq \lk$ between them.
Since $c_1$ and $c_2$ are part of the same component we are only allowed 
to assign one variable to reduce the backdoor depth of that component to $k-1$.
No matter which variable we choose, since $c_1$ and $c_2$ are disjoint,
one of them will continue to exist in the reduced graph, which will then have 
backdoor depth at least $k-1$ and again, 
the existence of $c_1$, $c_2$, and $P$ implies that $\sbdtwz(G) \geq k$.

\begin{figure}[H]
\centering
   \scalebox{0.7}{\input{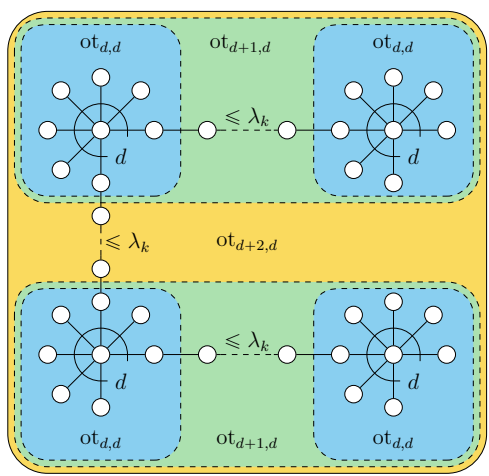}}
   \caption{Schematic depiction of an obstruction-tree for maximal clause degree $d$ 
   and strong recursive backdoor depth $d+2$.
   Here, the notation ``$\text{ot}_{i,d}$'' stands for the notion of $(i,d,k)$-obstruction-tree, as formally defined in \cref{def-obs-tree}.}
   \label{fig:proof_sketch}
\end{figure}

Given a maximal clause degree $d$, 
we generalize this strategy for arbitrary depths $k = d + j$ 
by searching for so called {\em obstruction-trees}.
An obstruction-tree for depth $k$ is a structured set of vertices, 
whose existence in $G$ will guarantee that $G$ has a strong recursive backdoor depth of at least $k$.
We start by searching for $d$-clauses as obstruction-trees for depth $d$. 
We search for obstruction-trees for depth $d + j + 1$ by searching for
two obstruction-trees of depth $d + j$
that have disjoint neighborhoods and
are connected by a path of bounded length.
A schematic depiction of obstruction-trees
for maximal clause degree $d$
and backdoor depths $d$, $d+1$, and $d+2$
is shown in \cref{fig:proof_sketch}.
Since the obstruction-trees are based on $d$-clauses we can
construct an fpt algorithm that either finds an obstruction-tree
or a small backdoor which reduces the maximal clause degree of $G$
to $d-1$.

The algorithm to find obstruction-trees, described in 
\cref{finding_ots_is_fpt}, is at the heart of our proof.
We use this algorithm to solve $G$
by recursively searching for obstruction-trees for depth $k + 1$.
In each round we either abort and conclude that $G$ has strong recursive backdoor depth
at least $k+1$ or reduce $d$ and recurse until we arrive at $\Cc_0$, 
where we can trivially check satisfiability.
If the graph splits into multiple components we can handle the components
separately and aggregate the results.


\section[Permissive SRB Detection Is FPT]{Permissive Strong Recursive Backdoor Detection to $\Cc_0$ Is FPT}

We start by formalizing and proving the observations made in \cref{sec:sketch}.


\begin{lemma}[Limited Clause Degree]\label[lemma]{limited_clause_degree}
  For every incidence graph $G$ and integer $d$, if $\sbdtwz(G) \le d$, then $G$ has a maximal clause degree at most $d$, i.e.\ $G\in \Cc_d$.
\end{lemma}

\begin{proof}
    Proof by induction on $d$.

    \inductionBaseCase
    $d = 0$. If $\sbdtwz(G) \le 0$, then $G \in \Cc_0$.


    \inductionStep
    Let $d$ be an integer and $G$ an incidence graph with $\sbdtwz(G) \le d+1$.
    Let $c$ be any clause in $G$.
    Let $H$ be the connected component of $c$ in $G$.
    By assumption, $\sbdtwz(H) \le d+1$,
    and there must be a variable vertex $x$ such that
    for every literal $x_\star$ we have $\sbdtwz(H[x_\star]) \le d$.
    By induction, we also have that $H[x_\star] \in \Cc_d$.

    \algCase{1}
    $x$ is not connected to $c$.
    Then $c$ is still intact in $H[x_\star]$ and by induction contains at most $d$ variables.

    \algCase{2}
    $x$ is connected to $c$ by an edge with polarity $+$.
    Then $c$ still exists in $H[x_-]$ and by induction has degree at most $d$ in $H[x_-]$. 
    Therefore $c$ contains at most $d+1$ variables in $G$.

    \algCase{3}
    The case that $x$ is connected to $c$ by an edge with polarity $-$ is 
    analogous to \emph{Case~2}.

    \noindent In all cases, $c$ contains at most $d+1$ variables in $G$, and therefore $G\in\Cc_{d+1}$.
\end{proof}


\begin{lemma}[Low Diameter]\label[lemma]{low_diameter}
    Let $G$ be a incidence graph.
    If either $\sbd_{\Cc_0}(G) \leq k$ or \mbox{$\wbd_{\Cc_0}(G) \leq k$}, then every connected component of $G$ has a diameter of at most $4 \cdot 2^k - 4$.
\end{lemma}

\begin{proof}
    Proof by induction on $k$.
    For brevity we write $\bd(G) \leq k$ for the fact that either $\wbd_{\Cc_0}(G) \leq k$ or $\sbd_{\Cc_0}(G) \leq k$.

    \algHeading{Base Case}
    $k = 0$.
    In this case, $G$ is edgeless, and the statement holds.


    \algHeading{Induction Step}
    Assume towards a contradiction that $\bd(G) \leq k+1$ and that
    $G$ has a connected component $H$ of diameter at least $4 \cdot 2^{k+1} - 3$.
    Hence $H$ contains two vertices connected by a shortest path
    $P=(v_1,\ldots, v_m)$
    with $ m= 4 \cdot 2^{k+1} - 2$ (such that if $v_i$ is a clause vertex, then
    $v_{i+1}$ is a variable vertex, and if $v_i$ is a variable vertex, then $v_{i+1}$
    is a clause vertex).
    Also there exists a literal $y_\star$ such that $y$ is a variable vertex in $H$ 
    witnessing that $\bd(G) \le k+1$.
    Since $P$ is a shortest path from $v_1$ to $v_m$, $y$ can only be connected to 	at most $2$ clauses in $P$ at distance 2 from each other.
    Let $v_{i-1}$ and $v_{i+1}$ be the two clauses connected to $y$ (the reasoning also works with only one or zero such $v_j$).
    Now assigning $y_\star$ can split $P$ by deleting $v_{i-1}$ and $v_{i+1}$.
    We then have that $G[y_\star]$ still contains $(v_1,...,v_{i-2})$
    and $(v_{i+2},...,v_m)$ as shortest paths.
    One of those two paths will include at least
    \[
    \left \lceil \frac{m - 3}{2} \right \rceil
    =
    \left \lceil \frac{4 \cdot 2^{k+1} - 2 - 3}{2} \right \rceil
    =
    \left \lceil 4 \cdot 2^{k} - \frac{5}{2} \right \rceil
    = 4 \cdot 2^{k} - 2
    \]
    vertices.
    One component of $H[y_\star]$ therefore has a diameter of at least $4 \cdot 2^{k} - 3$.
    This contradicts the fact that, by induction, $\bd(H[y_\star])\le k$. Therefore we have $\bd(G)\le k+1$.
\end{proof}
In the rest of the paper, we use $\lk = 4 \cdot 2^k$ for brevity.

\subsection{Obstruction-Trees}


We now turn to the concept of obstruction-trees. We first define them and then
prove some of their properties.
The main property, proved in \cref{ots_obstruct}, is that the existence of such trees witnesses
a lower bound for the depth of a strong recursive backdoor to the class~$\Cc_0$.

\begin{definition}[Obstruction-Trees and Destroy Neighborhoods]\label[definition]{def-obs-tree}
  For all integers $k,d$ and incidence graphs $G$ in $\Cc_d$, we inductively for $i \geq d$ define 
  the notion of an {\em $(i,d,k)$-obstruction-tree} $T$ of $G$
  with {\em elements} $\el(T)$ 
  and {\em destroy-neighborhood}~$\nd G T$.
  We use $\cla(T)$ and $\var(T)$ to denote the clauses and variables of $\el(T)$.

  \medskip
  \noindent For a set $T = \{c,x_1,...,x_d\}$, where $c$ is a $d$-clause of $G$ 
  with $\var(c)=\{x_1,\ldots,x_d\}$,
  we have:
  \begin{enumerate}
    \item $T$ is a $(d,d,k)$-obstruction-tree.
    \item $\el(T)$ are the elements of $T$.
    \item $\nd{G}{T} := \var(T)=\{x_1,\ldots, x_d\}$.
  \end{enumerate}

  \noindent Inductively, for a triple $T=(T_1,P,T_2)$ where $T_1$ and $T_2$
  are $(i,d,k)$-obstruction-trees of $G$
  such that $\nd G {T_1} \cap \nd G {T_2} = \emptyset$, and $P$ is a path of length at most $\lambda_k$ connecting~$T_1$ and~$T_2$, we have:
  \begin{enumerate}
    \item $T$ is an $(i+1,d,k)$-obstruction-tree.
    \item $\el(T) := \el(T_1) \cup \el(P)\cup \el(T_2)$.
    \item $\nd G{T} := \var(T) \cup \{x : \text{there exist } c_1,c_2 \in \cla(T)\text{with $\{x,c_1\}\in E_+$ and $\{x,c_2\}\in E_-$}\}$.
  \end{enumerate}
\end{definition}

We now show, that our definition of a destroy neighborhood is a small set of variables, 
that shields the obstruction-tree from the rest of the graph.
Remember that $\lk = 4\cdot 2^k$.

\begin{proposition}[$N^\dagger$ Is Small]\label[proposition]{nd_is_small}
  For all integers $i,d,k$ with $d \leq k$, for every incidence graph $G$ in $\Cc_d$, and every $(i,d,k)$-obstruction-tree $T$ of $G$,
  we have $|\el(T)| \le 3^{i-d}\cdot \lk$ and $|\nd{G}{T}| \leq 3^{i - d} \cdot \lk \cdot d$.
\end{proposition}

\begin{proposition}[$N^\dagger$ Is a Destroy Neighborhood]\label[proposition]{nd_is_a_destroy_neighborhood}
    For all integers $i,d,k$,
    every incidence graph $G$ in $\Cc_d$,
    every $(i,d,k)$-obstruction-tree $T$ of $G$,
    and every variable $x$ of $G$, 
    if $x \notin \nd{G}{T}$, then $T$ is also
    an $(i,d,k)$-obstruction-tree in at least one of $G[x_+]$ and $G[x_-]$.
\end{proposition}
Due to space constraints, both proofs were moved to \cref{nd_is_small_proof}
and \cref{nd_is_a_destroy_neighborhood_proof}.
We now turn to the main property of obstruction-trees, 
which explains why this notion is relevant in this context.

\begin{proposition}[Obstruction-Trees Obstruct]\label[proposition]{ots_obstruct}
    For all integers $i,d,k$ and every incidence graph $G$ in $\Cc_d$,  
    if there is an
    $(i,d,k)$-obstruction-tree $T$ of $G$, then $\sbdtwz(G) \geq i$.
\end{proposition}

\begin{proof}
    Proof by induction on $i$.

    \inductionBaseCase
    $i = d$.
    Follows immediately from \cref{limited_clause_degree}.


    \inductionStep
    Let $T$ be an $(i+1,d,k)$-obstruction-tree of $G$, and assume towards a contradiction
    that $\sbdtwz(G) < i+1$.
    By \cref{def-obs-tree} we get $T_1$ and $T_2$, two $(i,d,k)$-obstruction-trees connected by a path $P$.
    Additionally in every connected component~$H$ of~$G$,
    $\sbdtwz(H) < i+1$ holds as well.
    Since $\el(T)$ is connected, there exists one component~$H$ containing~$T$.
    Then, there must exists a variable $x$ in $H$ 
    such that $\sbdtwz(H[x_+]) < i$ and $\sbdtwz(H[x_-]) < i$.
    Assume $x \notin \nd{H}{T_1}$.
    Then by \cref{nd_is_a_destroy_neighborhood} we get that either~$T_1$ remains an $(i,d,k)$-obstruction-tree in either
    $H[x_+]$ or $H[x_-]$.
    We use the induction hypothesis to conclude that one of the graphs
    has strong recursive backdoor depth at least $i$ and we get a contradiction.
    Assume $x \in \nd{H}{T_1}$.
    Then $x \notin \nd{H}{T_2}$ by \cref{def-obs-tree} and we can make the same argument.
\end{proof}
Finally, and for technical reasons, we need to show that if we assign a variable,
we do not increase the strong recursive backdoor depth.
Also if we find an obstruction-tree after assigning some variables,
then it is also an obstruction-tree in the original graph with the same destroy neighborhood. 
\begin{lemma}[srbd$_{\Cc_0}$ Is Closed Under Assignments]\label[lemma]{srbd_closed_under_assignments}
    For every integer $k$, every incidence graph~$G$, and  every literal $x_\star$ in $G$, 
    if $\sbdtwz(G) \leq k$, then $\sbdtwz(G[x_\star]) \leq k$.
\end{lemma}
\begin{proposition}[Obstruction-Trees Can Be Lifted] \label[proposition]{ots_can_be_lifted}
    For all integers $i,d,k$ with $d \leq i$ and $d \leq k$, every incidence graph $G$ in $\Cc_d$,
    every obstruction-tree $T$ and every literal $x_\star$ of $G$, 
    if~$T$ is an $(i,d,k)$-obstruction-tree of $H := G[x_\star]$, then
    it is also an $(i,d,k)$-obstruction-tree of $G$ and we have $\nd{G}{T} = \nd{H}{T}$.
\end{proposition}
The proofs of these statements can be found in  \cref{srbd_closed_under_assignments_proof} and \cref{ots_can_be_lifted_proof}.

\subsection{Algorithms}

We finally turn to the algorithm. 
We first show that we can efficiently compute an
obstruction-tree, or make progress towards computing a strong recursive backdoor to $\Cc_0$.
Making progress here means decreasing the clause degree of the graph.
At every step, the algorithm may stop if it concludes that $\sbdtwz(G)>k$.


\begin{proposition}[Obstruction-Trees Are Easy to Compute]\label[proposition]{finding_ots_is_fpt}
  There is an algorithm that, given three integers $i,d,k$, with $d\le i\le k+1$,
  and an incidence graph $G$ in $\Cc_d$, in time $2^{2^O(k)}\cdot |G|$ either
  \begin{enumerate}
      \item returns an $(i,d,k)$-obstruction-tree $T$, or
      \item returns a strong recursive backdoor $B$ to $\mathcal{C}_{d-1}$ of depth at most $g(i,d,k):=3^{i-d} \cdot \lk \cdot d$, or
      \item concludes that $\sbdtwz(G) > k$.
  \end{enumerate}
\end{proposition}

\begin{proof}
  We fix $d, k$ and prove the claims by induction on $i$ and $|G|$.

  \inductionBaseCase
  When $i=d$, we search for a clause $c$ connected to $d$ variables.
  If we find $c$, then~$c$ is a $(d,d,k)$-obstruction-tree and we return it.
  If there is no such clause, then $G$ is also in~$\Cc_{d-1}$,
  and the leaf node labeled $G$ is a strong recursive backdoor to $\mathcal{C}_{d-1}$.

  \algHeading{Induction Step on $i$}
  We now fix $i$ and assume that have a working algorithm with parameters $(i,d,k)$ for any incidence graph $G \in \Cc_d$.
  We now explain by induction on $|G|$ how to build an algorithm with parameters $(i+1,d,k)$.

  \inductionBaseCase
  In the base case $|G|=1$,
  then $G\in\Cc_0$, and there is nothing to do.

  \algHeading{Induction Step on $|G|$ when $G$ is not connected}
  All connected components of $G$ have size strictly smaller than $G$,
  and we can run the algorithm with parameters $(i+1,d,k)$ on each.
  If in one connected component $H$, we find an $(i+1,d,k)$-obstruction-tree $T$,
  then $T$ is also an $(i+1,d,k)$-obstruction-tree of $G$ and we are done.
  If for one connected component $H$ we have $\sbdtwz(H) > k$,
  then it also holds that $\sbdtwz(G) > k$.
  Finally, if for every connected component $H$ we find
  a strong recursive backdoor $B_H$ to $\mathcal{C}_{d-1}$ of depth at most $g(i+1,d,k)$,
  we can merge them to build a recursive backdoor $B$ to $\mathcal{C}_{d-1}$ for $G$.
  In order to do this, we insert a root node labeled $G$, whose children are all the $B_H$.
  Since we do not insert a variable node, $B$ has still depth at most $g(i+1,d,k)$.

  \algHeading{Induction Step on $|G|$ when $G$ is connected}
  In this case, we use the induction hypothesis on $G$ with parameters $(i,d,k)$.
  If the algorithm provides a strong recursive backdoor or concludes that
  $\sbdtwz(G) > k$, we are done. We focus on the case where the algorithm returns an $(i,d,k)$-obstruction-tree $T$
  for $G$ and $\nd G T$.
  We define $\Tt$ as the set of all possible truth assignments to the variables of $\nd G T$.
  For every $\tau$ in $\Tt$, we define $H_\tau:= G[\tau]$.
  On every $H_\tau$ we run the algorithm given by induction
  with parameters $(i,d,k)$.

  \algCase{1}
  For at least one $H_{\tau}$
  we have that $\sbdtwz(H_{\tau})>k$.
  By \cref{srbd_closed_under_assignments}, $\sbdtwz(G)>k$ follows.

  \algCase{2}
  For at least one $H_{\tau}$
  we find an $(i,d,k)$-obstruction-tree $T_2$.
  By \cref{ots_can_be_lifted}, $T_2$ is also an $(i,d,k)$-obstruction-tree in $G$ and $\nd{H_\tau}{T_2} = \nd{G}{T_2}$.
  Since none of the variables in $\nd G {T_1}$ appear in $H_\tau$,
  we have that $\nd G{T_1}\cap \nd G{T_2} = \emptyset$.
  We run a BFS algorithm to find a shortest path $P$ between a vertex of $T_1$ and $T_2$ in $G$.
  If $P$ has length greater than $\lk$,
  we can conclude that 
  $\sbdtwz(G)>k$ by \cref{low_diameter}.
  If $P$ has length at most $\lk$,
  we have that $(T_1,P,T_2)$ is an $(i+1,d,k)$-obstruction-tree in $G$.

  \algCase{3}
  For every $H_{\tau}$
  we find a strong \recursiveBackdoor $B_{\tau}$
  for to $\Cc_0$ of depth at most $g(i,d,k)$.
  In this case, we can combine them and build a strong \recursiveBackdoor $B$ of $G$ to $\mathcal C_{d-1}$.
  To do so, we start with the complete binary tree, where at each step we branch over one variable in $\nd G {T_1}$.
  At depth $|\nd{G}{T_1}|$, each node corresponds to an assignment $\tau$ of $\Tt$.
  We then finish the tree by plugging 
  $B_\tau$ in the branch corresponding to $\tau$.
  $B$ is a strong \recursiveBackdoor of $G$ to $\mathcal C_{d-1}$, and its depth is bounded by:
  $|\nd G {T_1}| + g(i,d,k) \leq g(i+1,d,k)$.

\algHeading{Time Complexity}
The proof for the time complexity can be found in \cref{finding_ots_is_fpt_tc}.
\end{proof}
We now use the result of \cref{finding_ots_is_fpt} sufficiently many times so that the degree of the input graph reaches $0$.
Again, at any point, the algorithm may stop and conclude that $\sbdtwz(G)>k$. 
Remember that $\lk=4\cdot 2^k$.
\begin{theorem}\label[theorem]{find_srb}
  There is an algorithm that, given as input an integer $k$ and an incidence graph $G$,
  in time $2^{2^{\bigO{k}}} \cdot |G|$ either:
  \begin{enumerate}
    \item returns a strong recursive backdoor of $G$ to $\mathcal \Cc_0$ of depth at most $3^{k} \cdot \lk \cdot k^2$, or
    \item concludes that $\sbdtwz(G) > k$.
  \end{enumerate}
\end{theorem}
\begin{proof}
  Let $d$ be the maximal degree of a clause in $G$. If $d>k$ conclude that $\sbdtwz(G) > k$ by \cref{limited_clause_degree}.
  Otherwise handle $G$ using induction on $d$ to search for a \SRB to $\Cc_0$ of
  depth at most $3^{k} \cdot \lk \cdot d^2$:

  \inductionBaseCase
  $G \in \Cc_0$ and the node labeled $G$ is a \SRB to $\Cc_0$ of depth $0$.

  \inductionStep
  $G \in \Cc_{d+1}$.
  Run the algorithm presented in  \cref{finding_ots_is_fpt} with parameters 
  \mbox{$(k+1,d+1,k)$} on $G$.
  If it concludes that $\sbdtwz(G) > k$, or returns a $(k+1,d+1,k)$-obstruction-tree,
  conclude that $\sbdtwz(G) > k$ by \cref{ots_obstruct}.
  If a \SRB $B$ is returned,
  then $B$ will have depth at most $3^{k+1-(d+1)} \cdot \lk \cdot (d+1) \leq 3^{k} \cdot \lk \cdot (d+1)$
  and every leaf of~$B$ will be labeled with a graph $H$ in $\Cc_{d}$.

  We then apply the algorithm given by the induction hypothesis to every $H$.
  If for one $H$ we get that $\sbdtwz(H) > k$,
  conclude that $\sbdtwz(G) > k$ by \cref{srbd_closed_under_assignments}.
  If for every $H$ we get a \SRB $B_H$ to
  $\Cc_0$ of depth at most  $3^{k} \cdot \lk \cdot d^2$,
  we use the results to build a \SRB to $\Cc_0$ for $G$.
  To do so, we replace the leaf labeled $H$ in $B$ with $B_H$ for every $H$.
  As a result, $B$ will be extended to be a \SRB
  for $G$ to $\Cc_0$ with depth at most $3^{k} \cdot \lk \cdot (d+1)^2$.

  \algHeading{Time Complexity}
  The proof for the time complexity can be found in \cref{find_srb_tc}.
\end{proof}

\begin{corollary}
    Given a formula $\phi$ and a parameter $k$
    there is an algorithm that solves $\textsc{SAT}(\sbdtwz)$ in time $2^{2^{\bigO{k}}} \cdot |\phi|$.
\end{corollary}

\begin{proof}
    We compute the satisfiability of $\phi$ in two steps.
    First run the algorithm given in \cref{find_srb} 
    with parameters $\phi$ and $k$ in time $2^{2^{\bigO{k}}} \cdot |\phi|$.
    If the algorithm concludes that $\sbdtwz(\phi) > k$ we are finished.
    Otherwise a \SRB with depth at most $3^{k} \cdot \lk \cdot k^2$ of $\phi$ to $\Cc_0$
    is returned.
    
    Second, we make use of the calculated \SRB by running the algorithm described in 
    \cref{sat_using_srbds}, to determine the satisfiability of $\phi$.
    The satisfiability of a formula in $\Cc_0$ can be checked in constant time:
    If it contains no clause, then all clauses are trivially satisfied. 
    If it contains at least one clause, then that clause is empty and unsatisfiable.
    Therefore the second step runs in time $\bigO{2^{3^{k} \cdot \lk \cdot k^2} \cdot |\phi|}$.
    Adding up the running times, we get a total time complexity of $2^{2^{\bigO{k}}} \cdot |\phi|$.
\end{proof}


\section[{WRB Detection Is W[2]-Hard}]{Weak Recursive Backdoor Detection to $\Cc_0$ Is $W[2]$-Hard}
In this section, we show that the parametrized problem of detecting a weak
recursive backdoor of depth $k$ to the class of edgeless graphs is $W[2]$-hard
when parametrized by $k$.

\newcommand{\WRCZBD}{\ensuremath{\textsc{WR-}\Cc_0\textsc{-BD}}}

\begin{theorem}\label[theorem]{wrbd_is_hard}
    $\textsc{Weak-Recursive-}\Cc_0\textsc{-Backdoor-Detection}$ is W[2]-hard.
\end{theorem}
Due to space constraints, the proof was moved to \cref{wrbd_is_hard_proof}.

\section{Conclusion}
We have proposed the new notions of \emph{strong} 
and \emph{weak recursive backdoors}, 
which exploit the structure of 
formulas that can be recursively split into independent 
parts by partial assignments.
Recursive backdoors are measured by their depth 
and can contain, even at bounded depth, 
an unbounded number of variables.
In our work we have focused on the tractable 
base class of empty formulas $\Cc_0$.
We have shown,
that detecting weak recursive backdoors to $\Cc_0$ is W[2]-hard.
Our main technical contribution is an fpt algorithm that 
detects strong recursive backdoors of bounded depth to $\Cc_0$.
Even for $\Cc_0$ this extends tractable SAT Solving 
to a new class of formulas.

Our result raises the question of whether the 
detection of strong recursive backdoors can be expanded
to larger base classes such as 2CNF or Horn.
Especially 2CNF seems to be in reach,
as similar to $\Cc_0$, incidence graphs of formulas 
with bounded recursive backdoor depth to 2CNF have a bounded clause degree.
This is the first ingredient for our algorithm to $\Cc_0$.
However our algorithm is limited to finding backdoors to $\Cc_0$, 
as the second ingredient, which is bounded incidence graph diameter, 
is not given when searching for backdoors to 2CNF.


\bibliography{ref}

\appendix

\section{Omitted Proofs}

\subsection[Proof of Proposition 5.4]{Proof of \cref{nd_is_small}}\label{nd_is_small_proof}
\begin{proof}
    The second claim easily follows the first one. The set $\nd{T}{G}$ contains $\var(T)$ and at most the variables connected
    to a clause from $\cla(T)$.
    Since $G \in \Cc_d$, we get that $|\nd{G}{T}| \leq |\el(T)| \cdot d$.
    Now we prove that $|\el(T)| \leq 3^{i-d} \cdot \lk$ by induction on $i$.

    \inductionBaseCase
    $T$ is an $(d,d,k)$-obstruction-tree. By definition, $|\el(T)| = |T| = d + 1 \leq 3^{d-d} \cdot \lk$.

    \inductionStep
    $T$ is an $(i+1,d,k)$-obstruction-tree.
    Then $T = (T_1,P,T_2)$ such that~$T_1$ and~$T_2$ are $(i,d,k)$-obstruction-trees of $G$
    and $|\el(T)| \leq |\el(T_1)| + |\el(P)| + |\el(T_2)|$.
    We apply our induction hypothesis to conclude that both
    $\el(T_1)$ and $\el(T_2)$ have at most $3^{i-d} \cdot \lk$ elements.
    $P$~has at most $\lk$ elements by definition.
    We conclude that $|\el(T)| \leq 3 \cdot 3^{i-d} \cdot \lk = 3^{i+1-d} \cdot \lk$.
\end{proof}

\subsection[Proof of Proposition 5.5]{Proof of \cref{nd_is_a_destroy_neighborhood}}\label{nd_is_a_destroy_neighborhood_proof}
In order to prove \cref{nd_is_a_destroy_neighborhood}, we first prove an intermediate result:
\begin{proposition}[Obstruction-Trees Are only Influenced by Adjacent Variables]\label[proposition]{ots_are_local}
    For all integers~$i,d,k$,
    every incidence graph $G$ in $\Cc_d$,
    every $(i,d,k)$-obstruction-tree $T$ of $G$,
    every variable $x$ in $G$,
    and every polarity $\star$, 
    if $x \notin \var(T)$ and for all $c \in \cla(T)$ we have $\{x,c\}\not\in E_\star$, then
    $T$ is still an $(i,d,k)$-obstruction-tree in $G[x_\star]$.
\end{proposition}
\begin{proof}
    Proof by induction on $i$.

    \inductionBaseCase
    $i = d$.
    Then $T$ contains a $d$-clause $c$ and its adjacent variables $\var(T)$.
    If~$x\notin \var(T)$, then $T$ remains untouched and continues to be a
    $(d,d,k)$-obstruction-tree of $G[x_\star]$.

    \inductionStep
    $i>d$.
    Then $T = (T_1,P,T_2)$, where $T_1$ and $T_2$ are $(i,d,k)$-obstruction-trees of $G$.
    Assume $x \notin \var(T)$ and for all $c \in \cla(T)$ we have $\{x,c\}\not\in E_\star$.
    Since $\cla(T_1)$ and $\cla(T_2)$ are both subsets of $\cla(T)$ we
    can apply our induction hypothesis and conclude that~$T_1$ and$T_2$
    are still $(i,d,k)$-obstruction-trees of $G[x_\star]$.
    Since $G[x_\star]$ is a subgraph of $G$
    we know that $\nd{G[x_\star]}{T_1} \cap \nd{G[x_\star]}{T_2}$ is still empty.
    Since $x$ is not contained in $\var(P) \subseteq \var(T)$ and is also not connected to a clause of $\cla(P) \subseteq \cla(T)$ by polarity $\star$,
    we conclude that $P$ still is a path of the same length in $G[x_\star]$.
    It follows that $T$ must be $(i+1,d,k)$-obstruction-tree in~$G[x_\star]$.
\end{proof}
We now continue with the proof of \cref{nd_is_a_destroy_neighborhood}:
\begin{proof}
    Assume towards a contradiction that $x \notin \nd{G}{T}$ and
    $T$ is no $(i,d,k)$-obstruction-tree of $G[x_+]$ and $G[x_-]$.
    If $T$ is no $(i,d,k)$-obstruction-tree in $G[x_+]$,
    then by \cref{ots_are_local}, either $x \in \var(T)$ or there exists a clause $c_1$
    connected to $x$ by a positive edge.
    Since the former contradicts with $x \notin \nd{G}{T}$, we have that $c_1$ exists.
    Now assume that $T$ is also no $(i,d,k)$-obstruction-tree in $G[x_-]$.
    By the same reasoning conclude that there exists a clause~$c_2$
    connected to $x$ by a negative edge.
    From the existence of both $c_1$ and $c_2$ connected with different polarities to $x$
    we conclude that $x \in \nd{G}{T}$ and get a contradiction.
\end{proof}


\subsection[Proof of Lemma 5.7]{Proof of \cref{srbd_closed_under_assignments}}\label{srbd_closed_under_assignments_proof}
\begin{proof}
    Proof by induction on $k$.

    \inductionBaseCase
    $k=0$. Then $G$ is edgeless and remains edgeless when a variable is assigned.


    \inductionStep
    Let $G$ be an incidence graph such that $\sbdtwz(G)\le k+1$, and let $x_\star$ be any literal of $G$.
    If $G$ is connected, then by \cref{def_sbd} there exists a variable $y$ such that
    $\sbdtwz(G[y_\star]) \leq k$.
    If $x = y$, then $\sbdtwz(G[x_\star]) \leq k+1$ holds trivially.
    If $x \neq y$ then we apply our induction hypothesis
    and because $\sbdtwz(G[y_\star]) \leq k$ get that $\sbdtwz(G[y_\star,x_\star])
    \leq k$. This leads to $\sbdtwz(G[x_\star,y_\star]) \le k$,
    which again implies that $\sbdtwz(G[x_\star]) \leq k+1$ holds.
    If $G$ contains multiple components,
    then the same argument applies for the component that contains $x$
    and the other components remain unchanged.
\end{proof}

\subsection[Proof of Proposition 5.8]{Proof of \cref{ots_can_be_lifted}}\label{ots_can_be_lifted_proof}
\begin{proof}
    Proof by induction on $i$.

    \inductionBaseCase
    $i = d$. Then $T$ contains a $d$-clause $c$ and its variables $x_1,\ldots, x_d$ in $H$ and $\nd{H}{T}$ contains all $x_i$.
    Since $G$ has maximal clause degree $d$, $c$ must also be a $d$-clause in $G$ with the same neighborhood.

    \inductionStep
    Assume $T$ is an $(i + 1,d,k)$ obstruction-tree of $H$.
    Then $T = (T_1,P,T_2)$ such that $T_1$ and $T_2$ are $(i,d,k)$-obstruction-trees of $H$, and
    $P$ is a path of length at most $\lk$.
    By applying the induction hypothesis, we get that $T_1$ and $T_2$
    are also $(i,d,k)$-obstruction-trees of $G$ and
    that $\nd{G}{T_1} = \nd{H}{T_1}$ and $\nd{G}{T_2} = \nd{H}{T_2}$ are disjoint.
    $P$ obviously still is a path of length at most $\lk$ in $G$,
    so $T$ is indeed an $(i + 1,d,k)$-obstruction-tree of $G$.

    We now show that $\nd{H}{T} = \nd{G}{T}$.
    Since $H$ is an induced subgraph of $G$,
    we get that $\nd{H}{T} \subseteq \nd{G}{T}$.
    To show $\nd{H}{T} \supseteq \nd{G}{T}$,
    pick any variable $y$ from $\nd{G}{T}$.
    If $y \in \var(T)$ we get that $y \in \nd{H}{T}$
    by definition.
    If $y$ is positively connected to $c_1$ and negatively connected to $c_2$
    for two clauses $c_1,c_2 \in \cla(T)$,
    then $y \neq x$, since otherwise the assignment 
    of~$y$ in $H$ would delete a clause from $T$, which
    contradicts the fact that $T$ is an $(i+1,d,k)$-obstruction-tree in $H$.
    Since $y$ is not equal to $x$, its edges to $c_1$ and $c_2$ are not affected by the assignment of
    $x$ in $H$ and again $y \in \nd{H}{T}$ holds.
    It follows that $\nd{G}{T} = \nd{H}{T}$.
\end{proof}

\subsection[Time Complexity of Proposition 5.9]{Time Complexity of \cref{finding_ots_is_fpt}}\label{finding_ots_is_fpt_tc}
\begin{proof}
Let us prove by induction that
the time complexity of the algorithm presented in \cref{finding_ots_is_fpt} is $2^{2^{O(k)}}\cdot |G|$.
This clearly holds when $|G| =1$, or when $i=d$. We now move on to the induction
and analyze the run of the algorithm with parameters $(i+1,d,k)$ on a graph $G$.

First, note that when we split among several connected components these components are disjoints.
The sum of the sizes of these components is the size of the $G$. Unifying the recursive backdoor,
given by the components, by adding a common root node takes at most linear time, which is consistent with our hypothesis.

Second, when the graph is connected we first run the algorithm with parameters $(i,d,k)$.
If the run does not stop there, we have an $(i,d,k)$-obstruction-tree $T$.
We then consider a number of truth assignments that is bounded by $2^{\nd G T}$.
For each of these assignment, we run again our algorithm with parameters $(i,d,k)$, on graphs smaller than $G$.

If we find a second $(i,d,k)$-obstruction-tree we then only need to compute shortest path, which can be performed in linear time.
If every truth assignment provides a backdoor-tree, plugging them together 
only takes time linear in the number of possible truth assignments.

All together the procedure stays linear and the constant factor gets multiplied by a factor of the form
$O\left(2^{\nd G T}\right)$ each time $i$ decreases by one. By \cref{nd_is_small}, this is bounded by $O\left(2^{3^{i-d}\cdot \lk \cdot d}\right)$.
Using that both $d\le k$ and $i\le k+1$, this is of the form $2^{2^{O(k)}}$.
As $i$ can decrease by one at most $i-d$ many times (and therefore at most $i$ many times) until we get to a base case,
we get that the final constant factor is of the form $\left(2^{2^{O(k)}}\right)^{i} = 2^{i 2^{O(k)}} = 2^{2^{O(k)}}$.

We finally have that the overall complexity of a run with parameters $(i,d,k)$ on a graph~$G$ is bounded by
$2^{2^{O(k)}}\cdot|G|$.
\end{proof}

\subsection[Time Complexity of Theorem 5.10]{Time Complexity of \cref{find_srb}}\label{find_srb_tc}
  \begin{proof}
  Let $f(k) \cdot |G|$ be the time complexity of \cref{finding_ots_is_fpt} and $g(k,d)$ be $3^k \cdot \lk \cdot d$.
  We prove the time complexity of $2^{g(k,d) \cdot d} \cdot f(k) \cdot |G|$ by induction on $d$.
  If $d=0$ then we only have to construct a single node, which can be done in constant time.
  For graphs in $\Cc_{d+1}$, running the algorithm of \cref{finding_ots_is_fpt} can be done in
  time $f(k) \cdot |G|$.
  If we do not find a \SRB, we can abort.
  Otherwise we find a \SRB of depth at most $g(k,d+1)$
  such that all its leaves are members of $\Cc_d$.
  We can apply our induction hypothesis and
  assume that for a single leaf $H$, we can finish in time $2^{g(k,d) \cdot d} \cdot f(k) \cdot |H|$.
  Since the sum of the number of vertices in all leafs of the backdoor is at most $2^{g(k,d+1)} \cdot |G|$,
  we get that full algorithm has a running time in
  \[f(k) \cdot |G| + 2^{g(k,d) \cdot d} \cdot f(k) \cdot 2^{g(k,d+1)} \cdot |G|
  \leq 2^{g(k,d+1) \cdot (d + 1)} \cdot f(k) \cdot |G|.
  \]
  Since both $2^{g(k,d) \cdot d}$ and $f(k)$ are in $2^{2^{\bigO{k}}}$, the overall time complexity of the algorithm is in
  $2^{2^{\bigO{k}}} \cdot |G|$.
\end{proof}

\subsection[Proof of Theorem 6.1]{Proof of \cref{wrbd_is_hard}}\label{wrbd_is_hard_proof}
\begin{proof}
    We are going to show the W[2]-hardness of \textsc{Weak-Recursive-}$\Cc_0$\textsc{-Backdoor-Detection} $(\WRCZBD)$
    by reduction from the W[2]-complete \textsc{Set Cover} problem 
    \mbox{\cite[Theorem 13.28]{cygan2015parameterized}}.
    An instance $I$ of the \textsc{Set Cover} problem is composed of
    a universe $U$, an integer $k$, and a set $S \subseteq P(U)$.
    $I$ is a yes-instance if there exists a subset $L$ of $S$ with size at most $k$,
    such that the union of all sets in $L$ is equal to $U$.

    We reduce $I= (S,U,k)$ to the {\WRCZBD} instance $(\phi, k+1)$, where
    $\phi$ is a CNF formula over the variables $\{b_1,...,b_{k+2},s_1,...,s_n\}$,
    where $n=|S|$. This formula is constructed in the following way:

    For each element of the universe $u \in U$ a corresponding clause $\sigma_u$ is created,
    which we call \emph{element-clauses}.
    For each set $S_i$ a corresponding variable vertex $s_i$ is created,
    which we call \emph{set-variables}.
    Then $s_i$ and $\sigma_u$ are connected by a positive edge when $u \in S_i$.
    Therefore in the incidence graph, $s_i$ dominates all the clauses whose corresponding elements are contained in $S_i$.

    In addition, we create $k+2$ fresh variables $b_i$.
    Each are individually and positively connected to a fresh clause $\beta_i$.
    Furthermore, all $b_i$ are negatively connected to all $\sigma_u$, creating a $K_{k+2,|U|}$ bi-clique.
    More formally, we have:
    \begin{align*}
        \beta_i := & b_i\\
        \sigma_{u} := & \bigvee_{i=1}^{k+2} \neg b_{i}  \vee \bigvee_{\{i\le n~:~u\in S_{i}\} }s_{i}\\
        \phi := & \bigwedge_{i=1}^{k+2} \beta_i \wedge \bigwedge_{u\in U}\sigma_{u}
    \end{align*}

    We will now prove that this is in fact a valid reduction.

\medskip\noindent
    $\Rightarrow$:
    If $I = (\{S_1,...,S_n\},U,k)$ is a yes-instance, there exists a set of indices $J\subseteq \{1,...,n\}$ of at most size $k$ such that $\bigcup_{j\in J} S_j=U$.
    We construct the weak recursive backdoor $\{s_{j+}| j\in J\}$ of size and depth at most $k$ that dominates all element clauses.
    Once every variable $s_j$ has been assigned, every  element-clause $\sigma_u$ has been satisfied. 
    Only the $k+2$ clauses $\beta_i$ remain. These clauses are disjoint.
    We can therefore complete the backdoor by adding at depth $k+1$ every literal $b_i$ simultaneously.

\medskip\noindent
    $\Leftarrow$:
    Let $I = (S, U, k)$ be an instance for the \textsc{Set Cover}, and assume that $I' = (\phi,k)$ is a yes-instance.
    Then there must exist a weak recursive backdoor of depth at most $k$ that reduces $\phi$ to an edgeless graph.
    In order to satisfy every $\beta_i$ clause, each of the $b_{i}$ literals must be contained in the backdoor.
    Therefore the size of the backdoor is at least $k+2$.
    Since the backdoor can have depth at most $k+1$,
    at least two of the $\beta_i$ clauses have to be split into disconnected components and satisfied separately at some point.

    Since every variable $b_i$ is connected to every element-clause $\sigma_u$,
    the graph only contains one connected component, as long as a clause $\sigma_u$ is not satisfied.
    Since the literals $\neg b_{i}$ cannot be part of the backdoor,
    there must be a set of only set-variables that when assigned satisfy every element-clause.
    In order to not exceed the recursion depth of $k+1$,
    that set must be of size at most $k$.
    Having a set of at most $k$ set variables satisfying every element-clause
    implies the existence of a set cover of $U$ of at most $k$ elements. So $I$ is a yes-instance.
\end{proof}

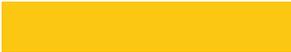
\end{document}